\documentclass[review]{elsarticle}

\usepackage{lineno}
\usepackage{hyperref}

\usepackage{algorithm} 
\usepackage{algpseudocode} 

\usepackage{mathrsfs}

\usepackage{amsfonts}

\usepackage{amssymb}

\usepackage{amsmath}

\usepackage{amsthm}
\newtheorem{theorem}{Theorem}
\theoremstyle{definition}

\theoremstyle{plain}
\newtheorem{assumption}{Assumption}
\newtheorem{lemma}[theorem]{Lemma}
\newtheorem{proposition}{Proposition}

\modulolinenumbers[5]

\journal{Journal of \LaTeX\ Templates}

\biboptions{authoryear}

\begin{document}

\begin{frontmatter}

\title{Approximating Posterior Predictive Distributions by Averaging Output From Many Particle Filters}

\author{Taylor R. Brown }
\address{Department of Statistics \\ University of Virginia}

\begin{abstract}
This paper introduces the {\it particle swarm filter} (not to be confused with particle swarm optimization): a recursive and embarrassingly parallel algorithm that targets an approximation to the sequence of posterior predictive distributions by averaging expectation approximations from many particle filters. A law of large numbers and a central limit theorem are provided, as well as an numerical study of simulated data from a stochastic volatility model.
\end{abstract}


\end{frontmatter}


\section{Introduction}

\subsection{Overview}

When parameter values are unknown, Bayesian forecasters seek to obtain the sequence of posterior predictive distributions. A common technique in practice consists of sampling parameter values from the posterior (either exactly or approximately), using each of these parameter values to calculate a conditional forecast, and then averaging those predictions together. Alternatively, one may use a mode or mean of the posterior distribution  (obtained, again, through either an approximate or exact calculation), and use that value to calculate a single prediction. 

When using nonlinear and/or non-Gaussian state-space models, there are two difficulties that come with this strategy. First, sampling from a posterior distribution is often computationally intensive \cite{mcmcforssms}. Second, despite particle filters having well-understood guarantees \cite{delmoral}, \cite{chopin2004}, \cite{douc2008}, \cite{kunsch2005}, \cite{legland2004}, \cite{VANHANDEL20093835}, \cite{Douc_2014}, \cite{whiteley2013}, the theoretical support for averaging randomly-instantiated particle filters, to the best of the author's knowledge, has not appeared previously in the literature.

This paper describes the {\it particle swarm filter} (not to be confused with particle swarm optimization): a recursive, embarrassingly parallel algorithm that targets an approximation to the sequence of posterior predictive distributions. It does this by averaging expectation approximations from $N_{\theta} \in \mathbb{N}$ particle filters, each with its $N_X \in \mathbb{N}$ state samples/particles, and it averages these particle filters with respect to the prior distribution, not the posterior distribution.

\subsection{Related Work}

The work in this paper is motivated by the desire for real-time forecasts. Accomplishing this requires, above all else, an on-line algorithm--one whose computational complexity remains bounded through time. This goal is longstanding, so I will mention how this work relates to other algorithms proposed in the past. Throughout this subsection it is important to be aware of where all these algorithms lie in the three-way tradeoff between bias, variance, and computational cost.

The first way to obtain real-time forecasts would be to use point estimates of the unknown parameter value to instantiate a single particle filter. After the single particle filter is instantiated, it can  be run in real-time to obtain forecasts at every time point. 

The drawbacks to this approach are obvious. First, if the point estimate summarizes a parameter posterior, forecasts obtained from the particle filter would only (roughly) target the posterior predictive distribution at the moment when the particle filter conditions on the same set of observed data that the parameter posterior does. Therefore, this approach comes with a bias. Second, sampling-based approaches for conducting inference on the parameter posterior can be computationally expensive. If these approaches are to be used, care must be taken that parameter estimates can be obtained in a timely manner. Third, using a single point estimate from the parameter posterior disregards parameter uncertainty.

If real-time forecasts from the posterior predictive distribution are to be obtained with a mixture of particle filters, then it seems that the parameter weights would need to change at every time point. The next question that one might ask is: “how can one continually sample from the sequence of parameter posteriors, and use those parameter samples to generate forecasts (via state filtering), all in an on-line way?” 

Accomplishing these two goals simultaneously, quantifying parameter uncertainty and forecasting, in a general state-space model with an on-line algorithm, might have first been attempted by \cite{selforganizing}. This novel idea was to extend the model's state space by including the unknown parameters as elements of each state vector. The state process of the extended model is still Markovian, and so, ostensibly, all particle filtering algorithms could be used on this same model after its state was reconsidered.

Unfortunately, the  issue of particle ``degeneracy" arises with this approach. As long as parameter samples are drawn at each time from parameter samples at previous time points, the number of unique parameter values is non-increasing, and this problem cannot be avoided. \cite{selforganizing} recommended sampling errors at each time point, for the parameter transitions, to mitigate this difficulty, and \cite{liuandwest} improves upon this approach by correlating said errors with state values. These errors, though, are artificial in the sense that one is no longer targeting the original model with the resulting particle filter. However, standard errors for estimates of the analogous quantities do not increase as quickly in time, and might all be bounded.

There exists another class of on-line algorithms used for this same task. These are only suitable for a narrower collection of state-space models: ones whose parameter posterior (this time conditioning on both state and observation trajectories) admits sufficient statistics that can be calculated recursively. This line of reasoning started with \cite{storvik} and \cite{fearnhead2002}, and was explored further in \cite{particlelearning}. These algorithms do not add any bias by refocusing on a different model, but simulation results suggest that these algorithms do not avoid the degeneracy issue.

Third, a common way to deal with the degeneracy issue is to use the ``resample-move" algorithm \cite{resamplemove}, or a related algorithm such as the $\text{SMC}^2$ algorithm described in \cite{smcsquared}. These algorithms are recursive. Unfortunately, however, they are not on-line, and so they are not suitable for real-time forecasting. A passage in this same paper reads ``a genuinely on-line analysis, which would provide constant Monte Carlo error at a constant CPU cost, with respect to all the components of [each time's state vector and the vector of parameters] may well be an unattainable goal."

Recapitulating and returning to the first approach, batch parameter estimation algorithms, when used in conjunction with on-line particle filters, could be useful in another way. After sampling from the parameter posterior, possibly with an MCMC algorithm such as one detailed in \cite{pmmh}, a collection of parameter samples could be used to instantiate many particle filters, instead of just one. All of these would still be run forward in time, and somehow each filter's forecasts could be averaged to obtain forecasts in real-time.

This approach would suffer from a similar bias issue, but it is still on-line; it quantifies parameter uncertainty in a potentially better way, and it would mitigate the degeneracy issue as it does not attempt to sample from the sequence of parameter posteriors. Furthermore, much of the additional computational cost could be parallelized away.

This is the central motivation for the particle swarm filter. In certain circumstances, one may not do better than sampling parameters once, and instantiating many particle filters with those different parameter values. Also, whether one is sampling parameters from an ``old" parameter posterior and forecasting for future time points, or if one is sampling from the parameter prior and generating forecasts immediately, both strategies may be seen as special cases of the same algorithm. 

Even though it may be simple to describe and program the averaging portion of this algorithm, it may not always be advisable to do so in practice. This paper’s primary contributions are Theorem~\ref{thm:pswarmconsistency} and Theorem~\ref{thm:pswarmclt1}, which can be used to assess whether and how to implement an averaging technique. 

Averaging and parallelization are not novel ideas in statistics. These have been explored deeply in previous work. For example, the particle swarm filter appears to be very similar to the $\text{IS}^2$ algorithm \cite{issquared}, which can also be used to average particle filter output. However, the purpose of the $\text{IS}^2$ algorithm is to estimate marginal likelihoods, not to provide forecasts. Second, parallelization procedures have been well-explored as well, for example in \cite{whiteley2016}. However, papers like this one are more concerned with ``scaling-up” particle filters by adjusting the resampling mechanism. They assume that the model’s parameters are known, and so are not concerned with the task of taking  into account parameter uncertainty to improve forecasts. 

The document is organized as follows: section~\ref{sec:1} describes the requisite background on state-space models and particle filters. Section~\ref{sec:2} provides a self-contained collection of technical results. This includes some well-known results concerning single particle filters, as well as two novel theorems regarding the particle swarm filter. Theorem~\ref{thm:pswarmconsistency} shows consistency of estimates at each time point, and Theorem~\ref{thm:pswarmclt1} shows asymptotic normality of estimates at each time point. Finally, section~\ref{sec:sim_study} provides a simulation study supporting the theoretical results.

\section{Definitions and Algorithms} \label{sec:1}

\subsection{General Notation}

For a measurable space $(\mathsf{Z}, \mathcal{Z})$, the space of measures and probability measures are denoted as $\mathscr{M}(\mathsf{Z})$ and $\mathscr{P}(\mathsf{Z})$, respectively. Product spaces are written with a numeral superscript (e.g. $\mathsf{Z}^2 = \mathsf{Z} \times \mathsf{Z}$), and product sigma fields' and measures' superscripts are written with the $\otimes$ character: $\nu \otimes \nu = \nu^{\otimes 2}$. Random variables will be referred to in both upper and lower case, and collections will be given subscripts that possess a colon (e.g. $\{z_t\}_{t=1}^T = z_{1:T}$). For any $\nu \in \mathscr{P}(\mathsf{Z})$ $L^p(\mathsf{Z}, \nu) = \{ f : \mathsf{Z} \to \mathbb{R} : \left( \nu(|f|^p) \right)^{1/p} < \infty  \}$, where $p=1,2$ will be used. 

For any bounded and measurable function $f$ defined on $\mathsf{Z}$, the supremum norm of this function is written as as $||f||_{\infty} = \sup_{z \in \mathsf{Z}} |f(z)|$. If $\nu \in \mathscr{M}(\mathsf{Z})$, then $\nu(f) = \int_{\mathsf{Z}} \nu(dz) f(z)$. For any (possibly-unnormalized) kernel $K : \mathsf{Z} \times \mathcal{Z} \to [0, \infty)$ and for any $A \in \mathcal{Z}$, we write the marginal measure $\int \nu(dz_1) K(z_1, A)$ as $\nu K(A)$, and the integral $\int K(z_1, dz_1)f(z_2)$ as $K(z_1,f)$. 

Finally, all random variables are assumed to be defined on some overarching probability space $(\Omega, \mathcal{F}, \mathbb{P})$.

\subsection{State-Space Models}

A state-space model is defined by a collection of probability distributions describing three things: an observed sequence of data $y_{1:T}$, an unobserved sequence of data $x_{1:T}$, and a collection of parameters $\theta$ governing all of the model's distributions \cite{iihmm}. Let $(\Theta, \mathcal{T})$, $(\mathsf{X}, \mathcal{X})$ and $(\mathsf{Y}, \mathcal{Y})$ be their measurable spaces, respectively. 

A prior distribution is selected for the unknown parameter $\pi : \mathcal{T} \to [0,1]$ for quantifying a priori parameter uncertainty. Next, a distribution for the state vector at the first time point is also required: $\mu_{\theta}(d x_1) := \mu(\theta, d x_1) : \Theta \times \mathcal{X} \to [0,1]$. The probabilistic time evolution of state vectors is described by the normalized state transition kernel $F_{\theta}(x_{t-1},dx_t) := F([\theta,x_{t-1}], dx_t ) : \Theta \times \mathsf{X} \times \mathcal{X} \to [0,1]$. Finally, conditioning on a state vector and the set of model parameters, the distribution of the contemporaneous observation is written as $G([\theta,x_{t}], dy_t): \mathsf{X} \times \mathcal{Y} \to [0,1]$. For convenience, this Markov kernel is assumed to be dominated by a sigma-finite measure on $(\mathsf{Y}, \mathcal{Y})$, written as $dy_t$. This allows the observation kernel to be written as $G([\theta,x_{t}], dy_t) = g_{\theta}(x_t,y_t)dy_t$. 

We also define the following unnormalized transition kernels $T_{\theta,y_t}(x_{t-1}, A) := T([\theta, x_{t-1},y_t], A) = \int_A F_{\theta}(x_{t-1},dx_{t}) g_{\theta,y_t}(x_t)$ and $T_{\theta,y_1}( A) := T([\theta,y_1], A) = \int_A \mu_{\theta}(dx_{1}) g_{\theta,y_1}(x_1)$. These will help with describing the classic particle filter algorithms in the next section.

\subsection{Quantities of Interest}\label{qoi}

When the parameter values of a state-space model are known, filtering distributions and evaluations of the likelihood are available. The likelihood is defined as 
\begin{equation*}
L_{\theta}(y_{1:t}) 
= 
\int \cdots \int \mu_{\theta}(dx_1) g_{\theta,y_1}(x_1) \times \cdots \times F_{\theta}(x_{t-1},dx_t) g_{\theta,y_t}(x_t).
\end{equation*}
In this paper, filtering distributions will be used to obtain expectations that help with forecasting. For example
\begin{equation}
\mathbb{E}[y_{t+1} \mid \theta, y_{1:t}] = \phi_{\theta, y_{1:t}}( f )
\end{equation}
where $f(x_t, \theta) = \mathbb{E}[y_{t+1} \mid x_t, \theta ]$. After deriving $f$ by hand, and plugging in parameter values (either assumed to be known or estimates), particle filters would provide recursive formulas for approximations to these quantities. Unfortunately, this technique is not able to quantify uncertainty with respect to the unknown parameter values--this process does not target the sequence of posterior predictive distributions. 

The fundamental idea for this paper is to use the above formula, but for many particle filters, and then average those predictions together. Each particle filter will use a parameter vector sampled from a user-chosen proposal distribution. The idea uses the following decomposition:
\begin{equation}\label{eqn:post_pred}
\mathbb{E}[y_{t+1} \mid y_{1:t}] = 
\int_{\Theta} \phi_{\theta, y_{1:t}}( f )p(d\theta \mid y_{1:t} ) \approx
\int_{\Theta} \phi_{\theta, y_{1:t}}( f )\pi(d\theta).
\end{equation}
Integration with respect to the posterior is approximated with integration with respect to the prior. Obtaining an accurate and recursive algorithm that targets the sequence of parameter posteriors, that possesses a computational cost that does not grow in time, has long been recognized as a difficult \cite{selforganizing}, \cite{liuandwest}, \cite{estimationreview} \cite{smcsquared}, and so that is why it is avoided. 

This approach, on the other hand, is not routinely biased above or below, and the prior can be easily replaced with a ``working prior." For example, let $1 \le h < t$ and replace the prior distribution in equation~\ref{eqn:post_pred} with an outdated posterior distribution:
\begin{equation}
\mathbb{E}[y_{t+1} \mid y_{1:t}] \approx
\int_{\Theta} \phi_{\theta, y_{1:t}}( f )p(d\theta \mid y_{1:h}).
\end{equation}

\subsection{Particle Filters}

When the parameter vector of a state-space model is known, a particle filter provides sample-based approximations to expectations with respect to each time point's filtering distribution $\phi_{\theta, y_{1:t}}(dx_t) : \Theta \times \mathsf{Y}^t \times \mathcal{X} \to [0,1]$. The filtering distributions for any state-space models always satisfies the following recursion

\begin{equation}
\phi_{\theta, y_{1:t}}(dx_t) = 
\frac{
\phi_{\theta,y_{1:t-1}}(T_{\theta,y_t}(\cdot, dx_t))
}{ 
\phi_{\theta,y_{1:t-1}}(T_{\theta,y_t}(\cdot, 1))
}.
\end{equation}
Unfortunately, these recursions are not always tractable for every state-space model, and so this explains the popularity and necessity of particle filter algorithms. 

To use one, the user must choose a sequence of proposal distributions. At the first moment, $Q_{\theta,y_1}(dx_1) := Q([\theta,y_1], dx_1) : \Theta \times \mathsf{Y} \times \mathcal{X} \to [0,1]$ is used to sample proposals, targeting the first time point's filtering distribution. These samples are weighted, and resampled from. At subsequent time points, $Q_{\theta,y_t}(x_{t-1},dx_t) := Q([\theta,x_{t-1}, y_t], dx_t) : \Theta \times \mathsf{X} \times \mathsf{Y} \times \mathcal{X} \to [0,1]$ is used to ``mutate" old particles into new ones.

A broad array of particle filtering algorithms fall under the category of Sequential Importance Sampling with Resampling (SISR) \cite{iihmm}. Turning samples that approximate $t-1$'s filtering distribution into samples that approximate time $t$'s filtering distribution is a two-step process. These steps go by the names mutation (or propagation) and resampling (or selection). 

In the mutation step, the proposal distribution is used to draw new samples from old ones, as well as to adjust each sample's weight according to how well they cohere with their new target. Time $t$ samples and weights are gathered into triangular arrays: $\{\tilde{X}^{N,j}_t\}_{1 \le j \le N_X}$ and $\{\tilde{W}_{1,\theta}^{N,j}\}_{1 \le j \le N_X}$, respectively. 

In the second stage, resampling transforms weighted samples into unweighted samples by drawing conditionally independently from what you possess. Algorithm~\ref{alg:sisr} describes the process in full detail. 
\begin{algorithm}
	\caption{SISR} 
	\label{alg:sisr}
	\begin{algorithmic}[1]
		\For {$t=1,2,\ldots$}
			\If{$t$ equals $1$} 
				\For {$j=1,2,\ldots,N_X$}
					\State Draw $\tilde{X}^{N,j}_1 \sim Q_{\theta,y_1}(dx_1)$
					\State Calculate $\tilde{W}_{1,\theta}^{N,j} = \frac{d T_{\theta,y_1}( \cdot)}{dQ_{\theta,y_1}( \cdot)}(\tilde{X}_1^{N,j})$
				\EndFor 
				\State Optionally calculate $\hat{\phi}_{N_X, \theta,y_{1}}(f)$ and/or $\hat{L}_{\theta}(y_{1})$
				\For {$j=1,2,\ldots,N_X$}
					\State Draw $I_1^{N,j}$ with probability $P(I_1^{N,j} = j') \propto \tilde{W}_{1,\theta}^{N,j'}$ 
					\State Set $X_1^j = \tilde{X}_1^{N,I_1^{N,j}}$ 
				\EndFor
			\Else  
				\For {$j=1,2,\ldots,N_X$}
					\State Draw $\tilde{X}_t^{N,j} \sim Q_{\theta,y_t}(x_{t-1}^{N,j},\cdot)$
					\State Calculate $\tilde{W}_{t,\theta}^{N,j} = \frac{d T_{\theta,y_t}(x_{t-1}^{N,j}, \cdot)}{dQ_{\theta,y_t}(x_{t-1}^{N,j}, \cdot)}(\tilde{X}_t^{N,j})$
				\EndFor
				\State Optionally calculate $\hat{\phi}_{N_X, \theta,y_{1:t}}(f)$ and/or $\widehat{\frac{L_{\theta}(y_{1:t})  }{L_{\theta}(y_{1:t-1}) }}$

				\For {$j=1,2,\ldots,N_X$}
					\State Draw $I_t^{N,j}$ with probability $P(I_t^{N,j} = j') \propto \tilde{W}_{t,\theta}^{N,j'}$ 
					\State Set $X_t^{N,j} = \tilde{X}_t^{N,I_t^{N,j}}$
				\EndFor 
			\EndIf
		\EndFor
	\end{algorithmic} 
\end{algorithm}

After mutation, two things can be calculated with the array of unnormalized weights. The first is an approximation to any filtering expectation:
\begin{equation}
\hat{\phi}_{N_X, \theta,y_{1:t}}(f) := \sum_{j=1}^{N_X} \frac{\tilde{W}_{t,\theta}^{N,j}}{\sum_{j'} \tilde{W}_{t,\theta}^{N,j'}} f(\tilde{X}_t^{N,j}).
\end{equation}

Second, one may approximate likelihoods using the following:
\begin{equation}
\hat{L}_{\theta}(y_{1}) = N_X^{-1}\sum_{j=1}^{N_X} \tilde{W}_{1,\theta}^{N,j} 
\hspace{10mm} 
\widehat{\frac{L_{\theta}(y_{1:t})  }{L_{\theta}(y_{1:t-1}) }} = N_X^{-1}\sum_{j=1}^{N_X} \tilde{W}_{1,\theta}^{N,j}.
\end{equation}

There is also the option of approximating expectations after resampling has been performed:
\begin{equation}
\check{\phi}_{N_X, \theta,y_{1:t}}(f) :=  N_X^{-1} \sum_{j=1}^{N_X} f(X_t^{N,j}).
\end{equation}
However, as shown in Theorem~\ref{thm:sisrclt}, this approximation has a larger asymptotic variance.

\subsection{Averaging Particle filters}

Particle filters are useful for state inference for a broad variety of state-space models. However, assuming that the parameters of a model are known is often not suitable.

Averaging the output from many different particle filters, all with randomly chosen static parameters, mitigates the effect that parameter uncertainty has on these estimates. Algorithm \ref{alg:swarm} describes the process whereby each particle filter is instantiated with a randomly chosen parameter vector, then run through the time series of observable quantities $y_{1:T}$, and at each point the predictions for the future are averaged. Also, no particle filter needs to communicate with any other, which facilitates parallel implementations. 

At the beginning of the algorithm, each parameter vector is sampled from a chosen distribution $\rho(d\theta) : \mathcal{T} \to [0,1]$. The outputs of all these particle filters are averaged together at each time point, with each summand being weighted according to an evaluation of the Radon-Nikodym derivative. 

\begin{algorithm} 
	\caption{Particle Swarm Filter} 
	\label{alg:swarm}
	\begin{algorithmic}[1]
		\For {$i=1,2,\ldots,N_{\theta}$}
			\State Draw $\theta^i \sim \rho(\cdot)$
			\State Calculate $\frac{d\pi}{d\rho}(\theta^i)$
		\EndFor
		\For {$t=1,2,\ldots$}
			\If{$t$ equals $1$}
				\For {$i=1,2,\ldots,N_{\theta}$}
				\For {$j=(i-1)N_X + 1 \le j \le i N_X$}
					\State Draw $\tilde{X}^{N,j}_1 \sim Q_{\theta^i,y_1}(dx_1)$
					\State Calculate $\tilde{W}_{1,\theta^i}^{N,j} = \frac{d T_{\theta^i,y_1}( \cdot)}{dQ_{\theta^i,y_1}( \cdot)}(\tilde{X}_1^{N,j})$
				\EndFor
					\State Optionally calculate $\hat{\phi}_{N_X, \theta^i,y_{1}}(f) $ and/or $\hat{L}_{\theta^i}(y_{1})$
				\For {$j=(i-1)N_X + 1 \le j \le i N_X$}
				
					\State Draw $I_1^{N,j}$ with probability $P(I_1^{N,j} = j') \propto \tilde{W}_{1,\theta^i}^{N,j'}$ 
					\State Set $X_1^j = \tilde{X}_1^{N,I_1^{N,j}}$ 
				\EndFor
				\EndFor
					\State Optionally calculate $\widehat{\pi\phi}_{y_{1}}(f)$ and/or $\widehat{\pi[L_{\theta}(y_{1})]}$			\Else  
				\For {$i=1,2,\ldots,N_{\theta}$}
				\For {$j=(i-1)N_X + 1 \le j \le i N_X$}				
					\State Draw $\tilde{X}_t^{N,j} \sim Q_{\theta^i,y_t}(x_{t-1}^j,\cdot)$
					\State Calculate $\tilde{W}_{t,\theta^i}^{N,j} = \frac{d T_{\theta^i,y_t}(x_{t-1}^{N,j}, \cdot)}{dQ_{\theta^i,y_t}(x_{t-1}^{N,j}, \cdot)}(\tilde{X}_t^{N,j})$
					\State Calculate $\hat{\phi}_{N_X, \theta^i,y_{1:t}}(f) := \sum_{j=(i-1)N_X + 1}^{i N_X} \frac{\tilde{W}_{t,\theta^i}^{N,j}}{\sum_{j'} \tilde{W}_{t,\theta^i}^{N,j'}} f(\tilde{X}_t^{N,j})$
				\EndFor
				\For {$j=1,2,\ldots,N_X$}
					\State Draw $I_t^{N,j}$ with probability $P(I_t^{N,j} = j') \propto \tilde{W}_{t,\theta^i}^{N,j'}$ 
					\State Set $X_t^{N,j} = \tilde{X}_t^{N,I_t^{N,j}}$
				\EndFor 
				\EndFor
				\State Optionally calculate $\widehat{\pi\phi}_{y_{1:t}}(f)$ and/or $\widehat{\pi[L_{\theta}(y_{1:t})]}$

			\EndIf
		\EndFor
	\end{algorithmic} 
\end{algorithm}

In its current form, equation~\ref{swarm_approx} provides approximations to expectations taken with respect to a sort of ``marginal filtering distribution." Following the strategy mentioned in subsection \ref{qoi}, using the particular $f(x_t,\theta) = \mathbb{E}[y_{t+1} \mid x_t, \theta ]$ provides forecasts according to an approximation of the posterior predictive distributions.

\begin{align}
\widehat{\pi\phi}_{y_{1:t}}(f) &:= N_{\theta}^{-1} \sum_{i=1}^{N_{\theta}} \frac{d\pi}{d\rho}(\theta^i) \hat{\phi}_{N_X, \theta^i, y_{1:t}}(f)  \label{swarm_approx} \\
\hat{\phi}_{N_X, \theta^i,y_{1:t}}(f) &:= \sum_{j=(i-1)N_X+1}^{iN_X} \frac{\tilde{W}_{t,\theta^i}^{N,j}}{\sum_{j'} \tilde{W}_{t,\theta^i}^{N,j'}} f(\tilde{X}_t^{N,j})
\end{align}

Even though it is not of central interest in this paper, it should be noted that it is also relatively easy to calculate an approximation to the marginal likelihood by pooling conditional likelihood estimates from each particle filter \cite{issquared}:
\begin{align}
\widehat{\pi[L_{\theta}(y_{1:T})]} &:= N_{\theta}^{-1} \sum_{i=1}^{N_{\theta}} \frac{d\pi}{d\rho}(\theta^i) \hat{L}_{\theta^i}(y_{1}) \prod_{t=2}^T \widehat{\frac{L_{\theta^i}(y_{1:t})  }{L_{\theta^i}(y_{1:t-1}) }}.
\end{align}


\section{Theoretical Results}\label{sec:2}

Theorems \ref{thm:pswarmconsistency} and \ref{thm:pswarmclt1} justify the use of this algorithm in the regime where $N_{\theta} \to \infty$ and $N_X \to \infty$. They require strong assumptions, and in particular, this algorithm would not be suitable in situations where the parameter space cannot be bounded. 

%
%

\subsection{Assumptions and Definitions}

This subsection provides some assumptions and definitions that are used in subsequent proofs. First, we define two filtrations that represent the information available at the end of the mutation step of time $t$, and the set of information available after the resampling step of time $t$, respectively. 

Before the algorithm starts, the available information is represented as the trivial sigma-field: $\mathcal{F}_{\theta,0}^N = \{ \Omega, \emptyset \}$. For $t \ge 1$ define 

\begin{equation}
\tilde{\mathcal{F}}_{\theta,t}^N = \mathcal{F}^N_{\theta,t-1} \bigvee \sigma \left( \tilde{X}_t^{N,1}, \ldots, \tilde{X}_t^{N,N_X} \right)
\end{equation}

and

\begin{equation}
\mathcal{F}_{\theta,t}^N = \tilde{\mathcal{F}}^N_{\theta,t} \bigvee \sigma \left(I_t^{N,1}, \ldots, I_t^{N,N_X} \right).
\end{equation}

These two filtrations should not be confused with the single sigma-field $\mathcal{F}_{\theta}^N$, which is the sigma-field generated by the $N_{\theta}$ parameter samples at the very beginning of the algorithm. This will be necessary only in the proof of Theorem~\ref{thm:pswarmclt1}.

Next are the assumptions used for all subsequent theorems. Broadly speaking, assumptions \ref{assmp1}-\ref{assmp5} are used for consistency and asymptotic normality results for individual particle filters, and the addition of assumptions \ref{assmp6}-\ref{assmp9} are required for results regarding algorithm~\ref{alg:swarm}. Proofs for all theorems are available in \ref{Appendix}.

Assumptions \ref{assmp1} and \ref{assmp2} restrict the denominators of fractions that must be positive. Assumption~\ref{assmp3} is used to show that $T_{\theta,y_t}(x_{t-1}, 1)$ is finite. Assumptions \ref{assmp4} and \ref{assmp5} are useful for dealing with the unnormalized weights found in both algorithms. Assumptions \ref{assmp6} and \ref{assmp7} allow us to give ``well-behaved" weights to each particle filter. 

Assumptions \ref{assmp8} and \ref{assmp9} are used to show consistency of algorithm~\ref{alg:swarm} by improving the convergence in probability to a {\it uniform} convergence in probability of each particle filter's estimate. These are used in conjunction with the theorem provided in \cite{ucip}, which is useful in other contexts such as, for example, estimating the parameters of nonlinear regression model.

\begin{assumption}{}\label{assmp1}
For any $\theta \in \Theta$, $\mu_{\theta}(g_{\theta,y_1})> 0$.
\end{assumption}

\begin{assumption}{}\label{assmp2}
For any $\theta \in \Theta$, any $x_{t-1} \in \mathsf{X}$, and all $t \ge 2$, $T_{\theta,y_t}(x_{t-1}, 1)> 0$ 
\end{assumption}

\begin{assumption}{}\label{assmp3}
For any $\theta \in \Theta$ and all $t \ge 1$, $||g_{\theta,y_t}||_{\infty} < \infty$.
\end{assumption}

\begin{assumption}{}\label{assmp4}
For any $\theta \in \Theta$, all $x_{t-1} \in \mathsf{X}$, $T_{\theta, y_1}(\cdot) \ll Q_{\theta,y_1}(\cdot)$ and $T_{\theta,y_t}(x_{t-1}, \cdot) \ll Q_{\theta,y_{t}}(x_{t-1}, \cdot)$ for $t \ge 2$.
\end{assumption}

\begin{assumption}{}\label{assmp5}
For any $\theta \in \Theta$, and for all $(x_{t-1},x_t) \in \mathsf{X}^{2}$, there exist positive versions of the two Radon-Nikodym derivative such that they are bounded. In other words

$$
\sup_{(x_{t-1},x_t)} \frac{d T_{\theta,y_t}(x_{t-1}, \cdot)}{dQ_{\theta,y_t}(x_{t-1}, \cdot)}(x_t) < \infty
$$

$$
\sup_{x_1} \frac{d T_{\theta,y_1}( \cdot)}{dQ_{\theta,y_1}( \cdot)}(x_1) < \infty
$$
\end{assumption}

\begin{assumption}{}\label{assmp6}
$\pi \ll \rho$.
\end{assumption}

\begin{assumption}{}\label{assmp7}
For any $\theta \in \Theta$, there exists a positive version of $\frac{d\pi}{d\rho}$ such that it's bounded. in other words,
$\sup_{\theta \in \Theta} \frac{d\pi}{d\rho}( \theta) < \infty$.
\end{assumption}

\begin{assumption}{}\label{assmp8}
The parameter space $\Theta$ is compact.
\end{assumption}

\begin{assumption}{}\label{assmp9}
For each $t \ge 1$, $\hat{\phi}_{N_X, \theta, y_{1:t}}(f)$ is stochastically equicontinuous.
\end{assumption}

It should be mentioned that \cite{ucip} provides several sufficient conditions for assumption~\ref{assmp9} to hold. It may also be shown on a case-by-case basis for any particular model of interest, but it appears to be difficult to verify in general modeling situations. Even for small changes in parameter values, the samples and normalized weights for each particle index can be quite large, making it difficult to bound above any absolute difference in two of these estimators. 

However, for two parameters $\theta_1, \theta_2 \in \Theta$, adding and subtract the common target and using the triangle inequality produces an upper bound for the absolute difference in two estimators:
\begin{equation}\label{expr:split}
|\hat{\phi}_{N_X, \theta_1,y_{1:t}}(f_2) - \phi_{\theta_2, y_{1:t}}(f_2)| + |\hat{\phi}_{N_X, \theta_2, y_{1:t}}(f_2) - \phi_{\theta_2, y_{1:t}}(f_2)|.
\end{equation}
Exponential inequalities \citep[Chapter-9]{iihmm} can then be brought to bear. A demonstration that it holds for the particular stochastic volatility model considered in section~\ref{sec:sim_study} is provided. 

Assumption~\ref{assmp8} is indeed restrictive. For any state-space model whose parameter space is noncompact, the use of algorithm~\ref{alg:swarm} will require informative priors.


\subsection{Consistency for SISR}

Theorem~\ref{thm:sisr} shows that approximations to expectations with respect to each time's filtering distribution are consistent. These are well-known results that are useful for state-space models whose parameters are known, and they will be used in the proofs to all subsequent results. The proof given in \ref{Appendix} is simply a rearrangement of results provided in \citep[Chapter-9]{iihmm}.

Theorem~\ref{thm:sisr} depends on lemmas \ref{lem:mutation} and \ref{lem:resampling}, which in turn depend on the primary workhorse \cite[Proposition 9.5.7]{iihmm}. A transcription of this is given in \ref{Appendix}. It is used to show that the resampling steps of algorithm~\ref{alg:sisr} preserve consistency, and that the mutation steps of algorithm~\ref{alg:sisr}, after changing the target expectations, also preserve consistency. In each application, the triangular array and filtration sequence it mentions are modified. 


\begin{lemma}\label{lem:mutation}
Under assumptions \ref{assmp1}-\ref{assmp4}, for any $\theta \in \Theta$ and $t \ge 2$, if $\check{\phi}_{N_X, \theta,y_{1:t-1}}(f')$ is consistent for any $f' \in L^1(\mathsf{X}, \phi_{\theta,y_{1:t-1}})$, then $\hat{\phi}_{N_X, \theta,y_{1:t}}(f)$ is consistent for any $f \in L^1(\mathsf{X}, \phi_{\theta,y_{1:t}})$.
\end{lemma}


\begin{lemma}\label{lem:resampling}
Under assumptions \ref{assmp4} and \ref{assmp5}, for any $\theta \in \Theta$, if $\hat{\phi}_{N_X, \theta,y_{1:t}}(f)$ is consistent for any $f \in L^1(\mathsf{X}, \phi_{\theta,y_{1:t}})$, then $\check{\phi}_{N_X, \theta,y_{1:t}}(f)$ is consistent for the same class of functions.
\end{lemma}

The following theorem uses lemmas \ref{lem:mutation} and \ref{lem:resampling} to prove inductively consistency for both estimators at each time step. 

\begin{theorem}\label{thm:sisr}
For any $\theta \in \Theta$, $t \ge 1$ and $f \in L^1(\mathsf{X}, \phi_{\theta,y_{1:t}})$, under assumptions \ref{assmp1}-\ref{assmp5},  

\begin{align*}
\hat{\phi}_{N_X, \theta,y_{1:t}}(f) &:= \sum_{j=1}^{N_X} \frac{\tilde{W}_{t,\theta}^{N,j}}{\sum_{j'} \tilde{W}_{t,\theta}^{N,j'}} f(\tilde{X}_t^{N,j}), \\
\check{\phi}_{N_X, \theta,y_{1:t}}(f) &:=  N_X^{-1} \sum_{j=1}^{N_X} f(X_t^{N,j})
\end{align*}
converge in probability to $\phi_{\theta,y_{1:t}}(f)$ as $N_X \to \infty$.

\end{theorem}


\subsection{Consistency for the Particle Swarm Filter}

Theorem~\ref{thm:pswarmconsistency} guarantees consistency of estimates of expectations taken with respect the ``marginal filtering distribution" when all nine assumptions hold. 

\begin{theorem}\label{thm:pswarmconsistency}
In algorithm~\ref{alg:swarm}, for any $t \ge 1$ and any $f \in L^1(\mathsf{X}, \phi_{\theta,y_{1:t}})$ such that $\phi_{\theta,y_{1:t}}(f) \in L^1(\Theta, \pi )$, under assumptions \ref{assmp1}-\ref{assmp9}, 
\begin{equation*}
\widehat{\pi\phi}_{y_{1:t}}(f) := N_{\theta}^{-1} \sum_{i=1}^{N_{\theta}} \frac{d\pi}{d\rho}(\theta^i) \hat{\phi}_{N_X, \theta^i, y_{1:t}}(f)  
\overset{\text{p}}{\to} 
\pi\phi_{y_{1:t}}(f)
\end{equation*}
as $N_X,N_{\theta} \to \infty$.
\end{theorem}

\begin{proof}[Proof of Theorem \ref{thm:pswarmconsistency}]
Let $\{\theta^i\}_{1 \le i \le N_{\theta}} \sim \rho$ denote the sample of parameter values used to instantiate $N_{\theta}$ particle filters. Using the triangle inequality, the overall estimator can be shown to be bounded above by the sum of two sequences that converge in probability to $0$:

\begin{align*}
0 &\le \left| \frac{1}{N_{\theta}} \sum_{i=1}^{N_{\theta}} \frac{d\pi}{d\rho}(\theta^i) \hat{\phi}_{N_X, \theta^i,y_{1:t}}(f) - [\pi \phi]_{y_{1:t}}(f) \right| \\
&\le \left|\frac{1}{N_{\theta}} \sum_{i=1}^{N_{\theta}} \frac{d\pi}{d\rho}(\theta^i) 
\hat{\phi}_{N_X, \theta^i,y_{1:t}}(f) 
- \frac{1}{N_{\theta} } \sum_{i=1}^{N_{\theta}} \frac{d\pi}{d\rho}(\theta^i) \phi_{\theta^i, y_{1:t}}(f) \right|  + \left| \frac{1}{N_{\theta} } \sum_{i=1}^{N_{\theta}} \frac{d\pi}{d\rho}(\theta^i)\phi_{\theta^i, y_{1:t}}(f)- [\pi \phi]_{y_{1:t}}(f) \right| \\
&\le \sup_{\theta \in \Theta} \frac{d\pi}{d\rho}(\theta) \sup_{\theta \in \Theta} \left| \hat{\phi}_{N_X, \theta,y_{1:t}}(f)  -   \phi_{\theta, y_{1:t}}(f) \right|  + \left| \frac{1}{N_{\theta} } \sum_{i=1}^{N_{\theta}} \frac{d\pi}{d\rho}(\theta^i) \phi_{\theta^i, y_{1:t}}(f)
- [\pi \phi]_{y_{1:t}}(f) \right|.
\end{align*}

$\sup_{\theta \in \Theta} \left| \frac{d\pi}{d\rho}(\theta) \right|$ is finite by assumptions \ref{assmp6} and \ref{assmp7}. The first summand then converges in probability to $0$ by \cite[Theorem 2.1]{ucip} and  assumptions \ref{assmp8}, \ref{assmp9} and Theorem \ref{thm:sisr}. The last term converges by the traditional weak law of large numbers, which holds because $\phi_{\theta,y_1}(f) \in L^1(\Theta,\pi)$, and because of assumption~\ref{assmp6}. 

Note that, in our application of \cite[Theorem 2.1]{ucip}, we have used the fact that $\phi_{\theta,y_{1:t}}(f)$ is stochastically equicontinuous. In this particular case, this boils down to traditional continuity in $\Theta$, owing to the fact that it does not rely on $N_X$.
\end{proof}


\subsection{Asymptotic Normality for SISR}


Lemmas \ref{lem:mutclt} and \ref{lem:resampclt} are used to prove Theorem~\ref{thm:sisrclt}, a useful result for using particle filters for state-space models with known parameters. This is another well-known result. The proof given in \ref{Appendix} is, again, just a rearrangement of results provided in \citep[Chapter-9]{iihmm}.

\begin{lemma}\label{lem:mutclt}
For any $\theta \in \Theta$ and $t \ge 2$, under assumptions \ref{assmp1}-\ref{assmp5}, if for any $f' \in L^2(\mathsf{X}, \phi_{\theta,y_{1:t-1}})$
\begin{equation}
N_X^{1/2}\left[\check{\phi}_{N_X, \theta,y_{1:t-1}}(f') - \phi_{\theta,y_{1:t-1}}(f') \right] \overset{\text{D}}{\to} 
\mathcal{N}\left(0, \sigma^2_{\theta,t-1}(f')\right), 
\end{equation}
as $N_X,N_\theta \to \infty$, and if $\check{\phi}_{N_X, \theta,y_{1:t-1}}(f'')$ converges in probability to $\phi_{\theta,y_{1:t-1}}(f'')$ for any $f'' \in L^1(\mathsf{X}, \phi_{\theta,y_{1:t-1}})$, then, for any $f \in L^2(\mathsf{X}, \phi_{\theta,y_{1:t}})$, $N_X^{1/2}\left[\hat{\phi}_{N_X, \theta,y_{1:t}}(f)- \phi_{\theta,y_{1:t}}(f) \right]$ is asymptotically normal with mean zero and variance 
\begin{equation}\label{eqn:asymp_var1}
\frac{\sigma^2_{\theta,t-1}\left( T_{\theta, y_{t}}  (x_{t-1},f)  \right) + \eta^2_{\theta,t-1}(f)}{ \left[ \phi_{\theta,y_{1:t-1}}(T_{\theta,y_t}(x_{t-1},1)) \right]^2 } ,
\end{equation}
where

\begin{align}\label{eta}
\eta^2_{\theta,t-1}(f) 
&:= 
\phi_{\theta,y_{1:t-1}} \left\{ \int Q_{\theta,y_t}(x_{t-1}, dx_t) \left( \frac{dT(x_{t-1}, \cdot)}{dQ(x_{t-1}, \cdot)}(x_t) \right)^2 f^2(x_t) \right\} \nonumber \\
&\hspace{2mm}-
\phi_{\theta,y_{1:t-1}} \left\{  T_{\theta,y_t}(x_{t-1},  f)^2 \right\}.
\end{align}
\end{lemma}


\begin{lemma}\label{lem:resampclt}
For any $\theta \in \Theta$, $t \ge 1$ and $f \in L^2(\mathsf{X}, \phi_{\theta,y_{1:t}})$, under assumptions \ref{assmp4} and \ref{assmp5}, if $N_X^{1/2} \left\{ \hat{\phi}_{N_X, \theta,y_{1:t}}(f) - \phi_{\theta,y_{1:t}}(f) \right\}$ is asymptotically normal with mean $0$ and variance $\tilde{\sigma}^2_{\theta,t}(f)$, and if $\hat{\phi}_{N_X, \theta,y_{1:t}}(f'') \overset{\text{p}}{\to} \phi_{\theta,y_{1:t}}(f'')$ as $N_X,N_\theta \to \infty$ for any $f'' \in L^1(\mathsf{X}, \phi_{\theta,y_{1:t}})$, then 
\begin{equation*}
N_X^{1/2}\left\{ \check{\phi}_{N_X, \theta,y_{1:t}}(f)- \phi_{\theta,y_{1:t}}(f) \right\}
\overset{\text{D}}{\to}
\mathcal{N}\left(0, \tilde{\sigma}^2_{\theta,t}(f) + \phi_{\theta,y_{1:t}}(f^2) - \left[ \phi_{\theta,y_{1:t}}(f) \right]^2\right).
\end{equation*} 
\end{lemma}


\begin{theorem}\label{thm:sisrclt}
For any $\theta \in \Theta$ and $t \ge 1$, under assumptions \ref{assmp1}-\ref{assmp5}, both

\begin{eqnarray*}
N_X^{1/2}\left\{ \hat{\phi}_{N_X, \theta,y_{1:t}}(f) -\phi_{\theta,y_{1:t}}(f)\right\}
\overset{\text{D}}{\to}
\mathcal{N}\left( 0 , \mathbb{V}_{\theta,1:t}(f)\right)
\\
N_X^{1/2} \left\{ \check{\phi}_{N_X, \theta,y_{1:t}}(f) - \phi_{\theta,y_{1:t}}(f)\right\}
\overset{\text{D}}{\to}
\mathcal{N}\left( 0 , \mathbb{V}^r_{\theta,1:t}(f) \right),
\end{eqnarray*}
where
\begin{eqnarray}
\mathbb{V}_{\theta,1:t}(f) = \frac{ \mathbb{V}_{\theta,1:t-1}(f) \left[ T_{\theta, y_{t}}  (x_{t-1},f)  \right] + \eta^2_{\theta,t-1}(f) }{ \left[ \phi_{\theta,y_{1:t-1}}(T(x_{t-1},1)) \right]^2 } \label{eqn:mut_asymp_var}\\
\mathbb{V}^r_{\theta,1:t}(f) = \mathbb{V}_{\theta,1:t}(f) + \phi_{\theta,y_{1:t}}(f^2) - \left[ \phi_{\theta,y_{1:t}}(f) \right]^2 \label{eqn:res_asymp_var} \\
\mathbb{V}_{\theta,1}(f) = \frac{ Q_{\theta,y_1}\left[ \left( \frac{dT_{\theta,y_1}}{dQ_{\theta,y_1}}(x_1) \right)^2 \left( f(x_1) - \phi_{\theta,y_1}(f) \right)^2  \right] }{ \left[T_{\theta,y_1}(1) \right]^2 }.\\
\end{eqnarray}

\end{theorem}

These asymptotic variance recursions show how and when accuracy is lost--expression~\ref{eqn:mut_asymp_var} functions as a law of total variance, and expression~\ref{eqn:res_asymp_var} is additive as the variance added by resampling is the same as the variance of the filtering distribution at that time point.


\subsection{A Central Limit Theorem for the Particle Swarm Filter}

\begin{theorem}\label{thm:pswarmclt1}
Under assumptions \ref{assmp1}-\ref{assmp9}, for any $t \ge 1$, and any $f \in L^2(\mathsf{X}, \pi \phi_{y_{1:t}})$, if $\mathbb{V}_{\theta,1:t}(f) < \infty$ for all $\theta \in \Theta$, then
\begin{align*}
&N_X^{1/2}N_{\theta}^{1/2}\left\{ \widehat{\pi\phi}_{y_{1:t}}(f) - \pi\phi_{y_{1:t}}(f) \right\} \overset{\text{D}}{\to} \\
&\mathcal{N}\left(0, 
\rho \left[\left(\frac{d\pi}{d\rho}(\theta)\phi_{\theta, y_{1:t}}(f)- [\pi \phi]_{y_{1:t}}(f) \right)^2 \right]
+
\pi\left[ 
\frac{d\pi}{d\rho}(\theta) \mathbb{V}_{\theta,1:t}(f)
\right]
\right)
\end{align*}
as $N_X,N_\theta \to \infty$.
\end{theorem}

\begin{proof}[Proof of Theorem~\ref{thm:pswarmclt1}]

Just as in the proof to Theorem~\ref{thm:pswarmconsistency}, write $\frac{1}{N_{\theta}} \sum_{i=1}^{N_{\theta}} \frac{d\pi}{d\rho}(\theta^i) 
\hat{\phi}_{N_X, \theta^i,y_{1:t}}(f) - [\pi \phi]_{y_{1:t}}(f)$ as $E_{N,\theta} + F_{N,\theta}$, where
\begin{equation}
E_{N,\theta} 
:= 
\frac{1}{N_{\theta}} \sum_{i=1}^{N_{\theta}} \frac{d\pi}{d\rho}(\theta^i) 
\hat{\phi}_{N_X, \theta^i,y_{1:t}}(f) 
- \frac{1}{N_{\theta} } \sum_{i=1}^{N_{\theta}} \frac{d\pi}{d\rho}(\theta^i) \phi_{\theta^i, y_{1:t}}(f)
\end{equation}
and
\begin{equation}
F_{N,\theta} := \frac{1}{N_{\theta} } \sum_{i=1}^{N_{\theta}} \frac{d\pi}{d\rho}(\theta^i)\phi_{\theta^i, y_{1:t}}(f)- [\pi \phi]_{y_{1:t}}(f).
\end{equation}

Using this fact, iterating the expectations, and applying the dominated convergence theorem, we can look at the joint characteristic function for these two pieces:
\begin{align*}
&\lim_{N_{\theta} \to \infty} \lim_{N_X \to \infty} 
\mathbb{E} \left[
\exp\left( is N_{\theta}^{1/2}N_X^{1/2} E_{N,\theta} \right)
\exp\left( it N_{\theta}^{1/2}N_X^{1/2} F_{N,\theta} \right) \right] \\
&=
\lim_{N_{\theta} \to \infty} \lim_{N_X \to \infty} 
\mathbb{E} \left[
\exp\left( it N_{\theta}^{1/2}N_X^{1/2} F_{N,\theta} \right) 
\mathbb{E} \left[
\exp\left( is N_{\theta}^{1/2}N_X^{1/2} E_{N,\theta} \right)
\mid \mathcal{F}_{\theta} \right] 
\right] \\ 
&= 
\exp\left( - \frac{t^2}{2} \rho \left[\left(\frac{d\pi}{d\rho}(\theta)\phi_{\theta, y_{1:t}}(f)- [\pi \phi]_{y_{1:t}}(f) \right)^2 \right]  \right)
\mathbb{E} \left[
\lim_{N_{\theta} \to \infty} \lim_{N_X \to \infty} 
\mathbb{E} \left[
\exp\left( is N_{\theta}^{1/2}N_X^{1/2} E_{N,\theta} \right)
\mid \mathcal{F}_{\theta} \right] 
\right] ,
\end{align*}
where the last line follows from a traditional central limit theorem applied to $\theta \mapsto \pi \phi_{y_{1:t}}(f)$. The second factor is equal to
\begin{align*}
&\lim_{N_{\theta} \to \infty} \lim_{N_X \to \infty} 
 \mathbb{E}\left[  \mathbb{E}\left[ \exp\left\{ is
N_X^{1/2}N_{\theta}^{1/2}\left\{  
 N_{\theta}^{-1}\sum_{i=1}^{N_{\theta}} \left\{ 
 \frac{d\pi}{d\rho}(\theta^i) \hat{\phi}_{N_X, \theta^i, y_{1:t}}(f)
- \frac{d\pi}{d\rho}(\theta^i) \phi_{\theta^i, y_{1:t}}(f)
\right\} \right\}
 \right\}
 \bigg\rvert \mathcal{F}_\theta^N \right] \right] \\
&=\lim_{N_{\theta} \to \infty} \lim_{N_X \to \infty} 
 \mathbb{E}\left[  \mathbb{E}\left[ \exp\left\{ is
\left\{  
 \sum_{i=1}^{N_{\theta}} N_{\theta}^{-1/2}\frac{d\pi}{d\rho}(\theta^i) N_X^{1/2} \left\{ 
  \hat{\phi}_{N_X, \theta^i, y_{1:t}}(f)
- \phi_{\theta^i, y_{1:t}}(f)
\right\} \right\}
 \right\}
 \bigg\rvert \mathcal{F}_\theta^N \right] \right] \\
&=\lim_{N_{\theta} \to \infty}  
 \mathbb{E}\left[   \prod_{i=1}^{N_{\theta}} \lim_{N_X \to \infty}
 \mathbb{E}\left[ \exp\left\{ \left(is N_{\theta}^{-1/2}\frac{d\pi}{d\rho}(\theta^i) \right)
\left\{  
  N_X^{1/2} \left\{ 
  \hat{\phi}_{N_X, \theta^i, y_{1:t}}(f)
- \phi_{\theta^i, y_{1:t}}(f)
\right\} \right\}
 \right\}
 \bigg\rvert \mathcal{F}_\theta^N \right] \right] \\
&=\lim_{N_{\theta} \to \infty}  
 \mathbb{E}\left[   \prod_{i=1}^{N_{\theta}}  
\exp\left( - \frac{s^2 }{2N_{\theta}}\left(\frac{d\pi}{d\rho}(\theta^i) \right)^2 \mathbb{V}_{\theta^i,1:t}(f) \right) 
 \right] \\
&=\lim_{N_{\theta} \to \infty}  
 \mathbb{E}\left[     
\exp\left( - \frac{s^2 }{2N_{\theta}}\sum_{i=1}^{N_{\theta}}	\left(\frac{d\pi}{d\rho}(\theta^i) \right)^2 \mathbb{V}_{\theta^i,1:t}(f) \right) 
 \right] \\
&=
\exp\left( - \frac{s^2 }{2} \pi\left[ 
\frac{d\pi}{d\rho}(\theta) \mathbb{V}_{\theta,1:t}(f)
\right]\right)   .
 \end{align*}

The last line follows because 
\begin{equation*}
N_{\theta}^{-1} \sum_{i=1}^{N_{\theta}}	\left(\frac{d\pi}{d\rho}(\theta^i) \right)^2 \mathbb{V}_{\theta^i,1:t}(f) 
\overset{\text{p}}{\to}
\pi\left[ 
\frac{d\pi}{d\rho}(\theta) \mathbb{V}_{\theta,1:t}(f),
\right]
\end{equation*}
and by the continuous mapping theorem, the exponentiated version of this converges in probability as well. By assumption~\ref{assmp7} and because each particle filter's asymptotic variance is bounded above, the order of the limit and the expectation can be changed due to the dominated convergence theorem. 

The above is a derivation of the asymptotic joint distribution of the random vector $\left(N_{\theta}^{1/2}N_X^{1/2}E_{N,\theta}, N_{\theta}^{1/2}N_X^{1/2}F_{N,\theta}  \right)^T$. The final result holds after an application of the delta method.

\end{proof}


\section{Numerical Experiments}\label{sec:sim_study}

Univariate time series data was simulated from the stochastic volatility model of \cite{taylor}. Assume $T \in \mathbb{N}^+$ and $t \in \mathbb{N} \cap [1, T]$, and let $w_{1:T}$ and $v_{1:T}$ be iid mean zero Gaussian random variables with variance $1$. The model is defined as
\begin{eqnarray*}
y_t = \beta \exp\left( \frac{x_t}{2} \right) v_t \\
x_t = \phi x_{t-1} + \sigma w_t, \hspace{2mm} t \ge 2 \\
x_1 = \frac{\sigma_x}{(1-\phi^2)^{1/2}} w_1.
\end{eqnarray*}
where $\phi = .91$, $\beta= .5$, and $\sigma = 1.0$. 

It is frequently assumed that the parameter vector $\theta = (\phi, \beta, \sigma)^\intercal \in (-1,1) \times (-\infty,\infty) \times [0,\infty)$; however, this parameter space is restricted further by the selection of an informative prior $\pi$. This provides a compliance with assumption~\ref{assmp9}. Independent uniform priors are chosen for these three parameters, each with with supports $[.5,.99]$, $[0.0, 1.0]$, and $[.5,2]$, respectively.

Using equation~\ref{eqn:post_pred} and derivations of $f_1(x_t, \theta) = \mathbb{E}[y_{t+1} \mid x_t, \theta] = 0$ and $f_2(x_t, \theta) = \mathbb{E}[y_{t+1}^2 \mid x_t, \theta] = \beta^2\exp\left(\phi x_t + \sigma^2/2 \right)$, we can plot the estimates of $\mathbb{E}[y_{t+1} \mid y_{1:t}]$ and (the approximate) $\mathbb{E}[y_{t+1}^2 \mid y_{1:t}]$ in time. Figure~\ref{fig:intervals} shows the observed time series $y_{1:1000}$, as well as plus or minus twice the approximate estimates of the posterior predictive forecast standard deviation. This plot was made with $1000$ parameter samples from $\rho$ set equal to $\pi$, and for each of those, $1000$ state particles in each particle filter.
\begin{figure}
\begin{center}
\includegraphics[scale=.75,bb=0 0 504 504]{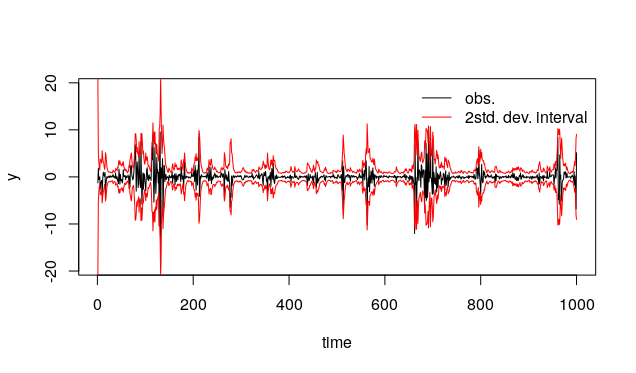}
\end{center}
\caption{Out-of-sample approximate posterior predictive forecasts.}\label{fig:intervals}
\end{figure}

To help visualize the uncertainty of the width of these intervals, the above process is repeated many times. $100$ sequences of the estimates of $\mathbb{E}[y_{t+1}^2 \mid y_{1:t}]$ are generated. Each sequence is calculated using $100$ parameter samples and $100$ particles for each particle filter. Figure~\ref{fig:uncertainty} shows the sample standard deviation of each time's\footnote{I remove the first time point's measurement of $447731.1$ because it is quite large (and expectedly so). Failing to do this produces a figure that is quite unreadable.} approximation to $\mathbb{E}[y_{t+1}^2 \mid y_{1:t}]$. This plot suggests that the uncertainty of these particular moment estimates from algorithm~\ref{alg:swarm} are bounded uniformly in time.
\begin{figure}
\begin{center}
\includegraphics[scale=.75,bb=0 0 504 504]{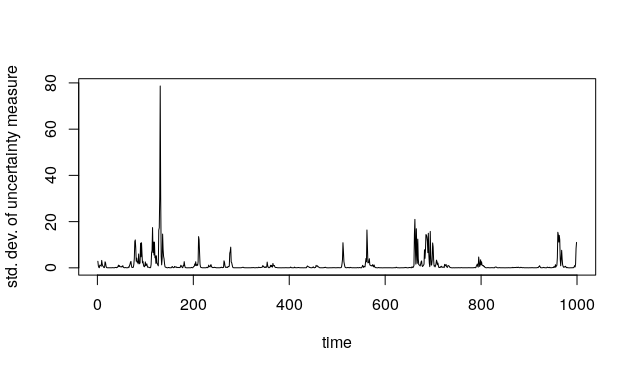}
\end{center}
\caption{Out-of-sample posterior predictive forecasts.}\label{fig:uncertainty}
\end{figure}

Assumption~\ref{assmp9} can be verified with any exponential deviation inequality; however, these typically require the assumption that $f_2$ is bounded. Even though we are assuming the parameter space is compact with assumption~\ref{assmp8}, it will not be the case in general because $\mathsf{X}$ is unbounded. For practical purposes here, though, we may pick a very large $M$ and swap $\bar{f}_2(x_t) = f_2(x_t) 1(|x_t| \le M)$ for $f_2$. As long as a large enough $M$ is chosen, it is extremely unlikely that there will be any difference in using these two functions in practice because $x_{1:1000}$ is a stationary process.

 
\section{Conclusion}\label{sec:conclusion}

I have presented an analysis of algorithm~\ref{alg:swarm}, the particle swarm filter. A central limit theorem has been deomonstrated, as well as a law of large numbers, justify the use of it in the regime where both $N_X$ and $N_{\theta}$ tend to infinity. Numerical experiments have demonstrated an application of this algorithm to estimating, in real-time, moments of the sequence of posterior predictive distributions while completely avoiding the need of any offline parameter estimation procedure, such as a MCMC algorithm.


\appendix
\section{Appendix}\label{Appendix}

\subsection{Consistency of Individual Particle Filters }

\begin{proposition}\label{prop957}
Let $\{V_{N,j}\}_{1 \le j \le M_N}$ be a triangular array of random variables and let $\{\mathcal{F}^N\}$ be a sequence of sub-$\sigma$-fields of $\mathcal{F}$. Assume that the following conditions hold true.

\begin{enumerate}
\item The triangular array is conditionally independent given $\{\mathcal{F}^N\}$ and for any $N$ and $j = 1,\ldots M_N$, $\mathbb{E}[|V_{N,j}| \mid \mathcal{F}^N] < \infty$.
\item The sequence $\{ \sum_{j=1}^{M_N} \mathbb{E}[|V_{N,j}| \mid \mathcal{F}^N]\}_{N \ge 0}$ is bounded in probability.
\item For any positive $\epsilon$
$$
\sum_{j=1}^{M_N} \mathbb{E}[ |V_{N,j}| 1\left(|V_{N,j}| \ge \epsilon \right)  \mid \mathcal{F}^N] \overset{\text{p}}{\to} 0.
$$
\end{enumerate}

Then
$$
\sum_{i=1}^{M_N} \{ V_{N,j} - \mathbb{E}[ V_{N,j} \mid \mathcal{F}^N] \overset{\text{p}}{\to} 0.
$$
\end{proposition}

\begin{proof}[Proof of Lemma~\ref{lem:mutation}]
For any $\theta \in \Theta$, the numerator of the right hand side of 

\begin{equation}\label{eqn:phihat}
\hat{\phi}_{N_X, \theta,y_{1:t}}(f) =
\frac{N_X^{-1} \sum_{j=1}^{N_X} \tilde{W}_{t,\theta}^{N,j} f(\tilde{X}_t^{N,j}) }{ N_X^{-1} \sum_{j'=1}^{N_X}\tilde{W}_{t,\theta}^{N,j'} }
\end{equation}

can be rewritten as

\begin{align}
&N_X^{-1} \sum_{j=1}^{N_X} \left\{ \tilde{W}_{t,\theta}^{N,j} f(\tilde{X}_t^{N,j})  -  \mathbb{E}\left[ \tilde{W}_{t,\theta}^{N,j} f(\tilde{X}_k^{N,j}) \mid \mathcal{F}_{\theta, t-1}^N \right] \right\}  \label{bigpicmutation} \\
&\hspace{1mm} + 
N_X^{-1} \sum_{j=1}^{N_X}  \mathbb{E}\left[ \tilde{W}_{t,\theta}^{N,j} f(\tilde{X}_t^{N,j}) \mid \mathcal{F}_{\theta, t-1}^N \right] \label{bigpicmutation2}.
\end{align}

We will start off by verifying the assumptions of proposition~\ref{prop957} in order to prove that the first of these two summands, the expression in \ref{bigpicmutation}, converges in probability to $0$.

First, the triangular array $\{ N_X^{-1} \tilde{W}_{t,\theta}^{N,j} f(\tilde{X}_t^{N,j}) \}_{1 \le j \le N_X}$ is conditionally independent given $\mathcal{F}_{\theta,t-1}^N$. This is true by the description of the mutation step of algorithm~\ref{alg:sisr}.

Second, for any $N \ge 1$ and $1 \le j \le N_X$, $N_X^{-1} \mathbb{E}\left[\tilde{W}_{t,\theta}^{N,j} f(\tilde{X}_t^{N,j}) \mid \mathcal{F}_{\theta,t-1}^N \right] < \infty$; and third: $\{ N_X^{-1} \sum_{j=1}^{N_X} \mathbb{E}\left[ \tilde{W}_{t,\theta}^{N,j} f(\tilde{X}_t^{N,j}) \mid \mathcal{F}_{\theta,t-1}^N \right] \}_{1 \le j \le N_X}$ is bounded in probability. These can be shown to true by some overlapping reasoning. 

Regarding these, $f \in L^1(\mathsf{X}, \phi_{\theta,y_{1:t}})$ implies

$$
0 \le \phi_{\theta,y_{1:t}}(|f|) 
= 
\frac{ \phi_{\theta, y_{1:t-1}}\left( T_{\theta, y_t}(x_{t-1}, |f|) \right) }{ \phi_{\theta, y_{1:t-1}}\left( T_{\theta, y_t}(x_{t-1}, 1) \right) } < \infty.
$$
The addition of assumptions \ref{assmp1}, \ref{assmp2} and \ref{assmp3} imply that $\left|T_{\theta,y_t}(x_{t-1}, 1)\right|= T_{\theta,y_t}(x_{t-1}, 1)$ in $L^1(\mathsf{X}, \phi_{\theta,y_{1:t-1}})$. This in turn implies  $T_{\theta,y_t}(x_{t-1}, |f|) = \left|T_{\theta,y_t}(x_{t-1}, |f|)\right|$ is in $L^1(\mathsf{X}, \phi_{\theta,y_{1:t-1}})$, which in turn implies that $T_{\theta,y_t}(x_{t-1}, f)$ is in there as well. Finally, the target expectation $\mathbb{E}\left[\tilde{W}_{t,\theta}^{N,j} f(\tilde{X}_t^{N,j}) \mid \mathcal{F}_{\theta,t-1}^N \right]$ is finite as well, because it is equal to $T_{\theta, y_t}(X_{t-1}, f)$ by assumption~\ref{assmp4}. 

Further, the sequence of averages in the third condition is bounded in probability because it converges in probability, due to our assumption of the consistency of $\check{\phi}_{N_X, \theta,y_{1:t-1}}(f')$, with $f'$ being set equal to $T(x_{t-1}, f)$.

Finally, for any $\epsilon > 0$, we have
\begin{equation}
N_X^{-1} \sum_{j=1}^{N_X} \mathbb{E}\left[ \tilde{W}_{t,\theta}^{N,j} f(\tilde{X}_t^{N,j}) 1 \left( N_X^{-1}\tilde{W}_{t,\theta}^{N,j} f(\tilde{X}_t^{N,j}) \ge \epsilon \right) \mid \mathcal{F}_{\theta,t-1}^N \right]
\overset{\text{p}}{\to} 0
\end{equation}
as $N_X,N_\theta \to \infty$. This is true by a dominated convergence argument that will also be used in the proof of Lemma~\ref{lem:resampling}.

We have verified all the assumptions of proposition~\ref{prop957}, so the sequence of averages in expression~\ref{bigpicmutation} converges in probability to $0$.

Assumption~\ref{assmp4} allows us to write the expression in expression~\ref{bigpicmutation2} as $ 
N_X^{-1} \sum_{j=1}^{N_X}  T(x_{t-1}^{N,j} ,f)$. This converges in probability to $\phi_{\theta,y_{1:t-1}}(T(x_{t-1}, f))$, as I have already shown $T(x_{t-1}, f) \in L^1(\mathsf{X}, \phi_{\theta,y_{1:t-1}}))$. 

Finally, the denominator of the right hand side of equation~\ref{eqn:phihat} converges to $\phi_{\theta,y_{1:t-1}} \left( T_{\theta,y_t}(X_{t-1},1) \right)$ because, as I have shown, $T_{\theta, y_t}(X_{t-1}, 1) \in L^1(\mathsf{X}, \phi_{\theta,y_{1:t-1}}))$. 

\end{proof}


\begin{proof}[Proof of Lemma~\ref{lem:resampling}]
Pick a $\theta \in \Theta$, $f \in L^1(\mathsf{X}, \phi_{\theta,y_{1:t}}))$, and write $\check{\phi}_{N_X, \theta,y_1}(f)$ as 

\begin{equation}
\left\{ \check{\phi}_{N_X, \theta,y_1}(f) - \hat{\phi}_{N_X, \theta,y_1}(f)\right\} + \hat{\phi}_{N_X, \theta,y_1}(f).
\end{equation}
The second term converges to the target by assumption. The resampling error, which is the first term in curly braces, converges to $0$ after an application of proposition~\ref{prop957}, whose assumptions are now verified.

The triangular array $\{ N_X^{-1} f(X_t^{N,j}) \}_{1 \le j \le N_X}$ is conditionally independent given $\tilde{\mathcal{F}}_{\theta,t}^N$ because of the description of the resampling step in algorithm~\ref{alg:sisr}. 

For any $N \ge 1$ and $1 \le j \le N_X$, $N_X^{-1} \mathbb{E}\left[  f(X_t^{N,j}) \mid \tilde{\mathcal{F}}_{\theta,t}^N \right] < \infty$ because

\begin{equation}
\mathbb{E}\left[  f(X_t^{N,j}) \mid \tilde{\mathcal{F}}_{\theta,t}^N \right] 
= 
\sum_{j=1}^{N_X} \frac{\tilde{W}_t^{N,j} }{ \sum_{j'}\tilde{W}_t^{N,j'}} f( \tilde{X}_t^{N,j})
\le ||f||_{\infty} < \infty.
\end{equation}
The sum of the normalized weights is $1$ by assumptions \ref{assmp4} and \ref{assmp5}. 

These same assumptions, along with the assumption of $\hat{\phi}_{N_X, \theta,y_{1:t}}(f)$'s consistency, also guarantee the third point of proposition~\ref{prop957}, that $\left\{ N_X^{-1} \sum_{j=1}^{N_X} \mathbb{E}\left[ f(X_t^{N,j}) \mid \tilde{\mathcal{F}}_{\theta,t}^N \right] \right\}_{1 \le j \le N_X}$ is bounded in probability. This is true because the sequence is consistent:

\begin{align*}
N_X^{-1} \sum_{j=1}^{N_X} \mathbb{E}\left[ f(X_t^{N,j}) \mid \tilde{\mathcal{F}}_{\theta,t}^N \right] 
&= \mathbb{E}\left[ f(X_t^{N,1}) \mid \tilde{\mathcal{F}}_{\theta,t}^N \right] \\
&= \sum_{j=1}^{N_X} \frac{\tilde{W}_t^{N,j} }{ \sum_{j'}\tilde{W}_t^{N,j'}} f( \tilde{X}_t^{N,j}) \overset{\text{p}}{\to} \phi_{\theta,y_{1:t}}(f).
\end{align*}

Finally, the fourth condition of proposition~\ref{prop957} is for any $\epsilon > 0$, 
\begin{equation}
N_X^{-1} \sum_{j=1}^{N_X} \mathbb{E}\left[ |f|(X_t^{N,j}) 1 \left( N_X^{-1} |f|(X_t^{N,j}) \ge \epsilon \right) \mid \tilde{\mathcal{F}}_{\theta,t}^N \right]
\overset{\text{p}}{\to} 0
\end{equation}
as $N_X,N_\theta \to \infty$. For any $0 \le C_{\epsilon} \le N_X \epsilon$,

\begin{equation*}
\mathbb{E}\left[ |f|(X_t^{N,j}) 1 \left( N_X^{-1} |f|(X_t^{N,j}) \ge \epsilon \right) \mid \tilde{\mathcal{F}}_{\theta,t}^N \right]
\le
\mathbb{E}\left[ |f|(X_t^{N,j}) 1 \left( |f|(X_t^{N,j}) \ge C_{\epsilon} \right) \mid \tilde{\mathcal{F}}_{\theta,t}^N \right].
\end{equation*}
For a fixed $C_{\epsilon}$, the right hand side converges in probability to $\phi_{\theta,y_{1:t}}(|f|1(|f| \ge C_{\epsilon}))$ as $N_X,N_\theta \to \infty$. This gives us

\begin{align*}
&\text{plim}_{N_X \to \infty} N_X^{-1} \sum_{j=1}^{N_X} \mathbb{E}\left[ |f|(X_t^{N,j}) 1 \left( N_X^{-1} |f|(X_t^{N,j}) \ge \epsilon \right) \mid \tilde{\mathcal{F}}_{\theta,t}^N \right] \\
&= 
\text{plim}_{C_{\epsilon} \to \infty} \text{plim}_{N_X \to \infty} N_X^{-1} \sum_{j=1}^{N_X} \mathbb{E}\left[ |f|(X_t^{N,j}) 1 \left( N_X^{-1} |f|(X_t^{N,j}) \ge \epsilon \right) \mid \tilde{\mathcal{F}}_{\theta,t}^N \right] \\
&\le
\text{plim}_{C_{\epsilon} \to \infty} \text{plim}_{N_X \to \infty} N_X^{-1} \sum_{j=1}^{N_X} \mathbb{E}\left[ |f|(X_t^{N,j}) 1 \left( |f|(X_t^{N,j}) \ge C_{\epsilon} \right) \mid \tilde{\mathcal{F}}_{\theta,t}^N \right] \\
&=
\text{plim}_{C_{\epsilon} \to \infty} \phi_{\theta,y_{1:t}}(|f|1(|f| \ge C_{\epsilon})) \\
&=  \phi_{\theta,y_{1:t}}( \lim_{C_{\epsilon} \to \infty} |f|1(|f| \ge C_{\epsilon})) \\
&= 0.
\end{align*}
The penultimate line follows from the dominated convergence theorem, which can be applied because $\phi_{\theta,y_{1:t}}\left[ |f| 1(|f| \ge C_{\epsilon})\right] \le \phi_{\theta,y_{1:t}}\left[ |f| \right] < \infty$.

\end{proof}


\begin{proof}[Proof of Theorem~\ref{thm:sisr}]

$\hat{\phi}_{N_X, \theta,y_{1}}(f)$ is consistent for $f \in L^1(\mathsf{X}, \phi_{\theta,y_{1}})$ by the traditional weak law of large numbers. It is also a corollary of \cite[Theorem 9.1.8]{iihmm}. 

Consistency of $\check{\phi}_{N_X, \theta,y_{1}}(f)$ follows from Lemma~\ref{lem:resampling}. For $t \ge 2$, consistency of $\check{\phi}_{N_X, \theta,y_{1:t}}(f)$ and $\hat{\phi}_{N_X, \theta,y_{1:t}}(f)$ arises from applying lemmas \ref{lem:mutation} and \ref{lem:resampling} inductively.
\end{proof}


\subsection{Asymptotic Normality for Individual Particle Filters}

Next, I restate \cite[Proposition 9.5.12]{iihmm} and provide two lemmas that guarantee its assumptions is here.
This result is not novel, but is included for self-containment

\begin{proposition}\label{prop9512}
Let $\{V_{N,i}\}_{1 \le i \le M_N}$ be a triangular array of random variables and let $\{\mathcal{F}^N\}_{N \ge 0}$ be a sequence of sub-$\sigma$-fields of $\mathcal{F}$. Assume that the following conditions hold true.
\begin{itemize}
\item The triangular array is conditionally independent given $\{\mathcal{F}^N\}_{N \ge 0}$, and for any $N$ and $i=1,\ldots,M_N$, $\mathbb{E}[V_{N,i}^2 \mid \mathcal{F}^N] < \infty$.
\item There exists a constant $\sigma^2 > 0$ such that
$$
\sum_{i=1}^{M_N}\left\{ \mathbb{E}\left[V_{N,i}^2 \mid \mathcal{F}^N \right] - \left(\mathbb{E}\left[V_{N,i} \mid \mathcal{F}^N \right]  \right)^2  \right\} \overset{\text{p}}{\to} \sigma^2.
$$
\item For all $\epsilon > 0$
$$
\sum_{i=1}^{M_N} 
\mathbb{E}\left[V_{N,i}^2 1\left( |V_{N,i}| \ge \epsilon \right)  \mid \mathcal{F}^N \right]  
\overset{\text{p}}{\to} 
0.
$$

Then for any $u$, 
\begin{equation}
\mathbb{E}\left[ \exp \left(i u \sum_{i=1}^{M_N}\left\{ V_{N,i}  - \mathbb{E}\left[V_{N,i} \mid \mathcal{F}^N \right]    \right\} \right)  \right] \overset{\text{p}}{\to} \exp\left(- \frac{u^2}{2}\sigma^2 \right).
\end{equation}
\end{itemize}
\end{proposition}

\begin{proof}[Proof of Lemma~\ref{lem:mutclt}]
Pick any $\theta \in \Theta$, $f \in L^2(\mathsf{X}, \phi_{\theta,y_{1:t}})$. Without loss of generality, assume that $\phi_{\theta,y_{1:t-1}}(T_{\theta,y_t}(x_{t-1}, f)) = 0$. 

The numerator on the right hand side of
\begin{align*}
N_X^{1/2}\left\{ \hat{\phi}_{N_X, \theta,y_{1:t}}(f) - \phi_{\theta,y_{1:t}}( f) \right\}
=
\frac{N_X^{-1/2} \sum_{j=1}^{N_X} \tilde{W}_{t,\theta}^{N,j} f(\tilde{X}_t^{N,j}) }{ N_X^{-1} \sum_{j'=1}^{N_X}\tilde{W}_{t,\theta}^{N,j'} }
\end{align*}
can be rewritten as $N_X^{1/2}\left\{ A_{N,\theta} + B_{N,\theta} \right\}$ where

\begin{align*}
A_{N,\theta} &:=  N_X^{-1} \sum_{j=1}^{N_X} \left\{ \tilde{W}_{t,\theta}^{N,j} f(\tilde{X}_t^{N,j})  -  \mathbb{E}\left[ \tilde{W}_{t,\theta}^{N,j} f(\tilde{X}_k^{N,j}) \mid \mathcal{F}_{\theta, t-1}^N \right] \right\}, \\
B_{N,\theta} &:= N_X^{-1} \sum_{j=1}^{N_X} \mathbb{E}\left[ \tilde{W}_{t,\theta}^{N,j} f(\tilde{X}_t^{N,j}) \mid \mathcal{F}_{\theta, t-1}^N \right] .
\end{align*}
By assumption~\ref{assmp4}, $N_X^{1/2} B_{N,\theta}$ is asymptotically normal as soon as we can show $T_{\theta, y_{t}}( \cdot, f)  \in L^2(\mathsf{X}, \phi_{\theta, y_{1:t-1} })$. This is indeed true--by Jensen's inequality
$$
\left[ \int F_{\theta}( x_{t-1} , dx_t)g_{\theta} (x_t,y_t) f(x_t) \right]^2 
\le  
\int F_{\theta}( x_{t-1} , dx_t)g^2_{\theta} (x_t,y_t) f^2(x_t) .
$$
Taking expectations on both sides and using assumptions \ref{assmp1} and \ref{assmp2}
\begin{align*}
\phi_{\theta, y_{1:t-1} } \left[\left( T_{\theta,y_{t}}( x_{t-1} , f) \right)^2 \right] 
&\le 
\phi_{\theta,y_{1:t-1}}[T(x_{t-1}, 1)]
\frac{\phi_{\theta, y_{1:t-1} } \left[\int F_{\theta}( x_{t-1} , dx_t)g^2_{\theta} (x_t,y_t) f^2(x_t) \right] }{\phi_{\theta,y_{1:t-1}}[T(x_{t-1}, 1)]}\\ 
&=
\phi_{\theta,y_{1:t-1}}[T(x_{t-1}, 1)] \phi_{\theta, y_{1:t} } \left[ g_{\theta} (x_t,y_t) f^2(x_t)\right].
\end{align*}
Assumption~\ref{assmp3} and the assumption of $f \in L^2(\mathsf{X}, \phi_{\theta,y_{1:t}})$ give an upper bounds for these factors because

\begin{enumerate}
\item $\phi_{\theta,y_{1:t-1}}[T(x_{t-1}, 1)] = \iint \phi_{\theta,y_{1:t-1}}(dx_{t-1}) F(x_{t-1}, dx_t)g_{\theta,y_t}(x_t) \le || g_{\theta,y_t}||_{\infty} $, 
\item $\phi_{\theta, y_{1:t} } \left[ g_{\theta} (x_t,y_t) f^2(x_t)\right] \le ||g_{\theta,y_t}||_{\infty} \phi_{\theta, y_{1:t} } \left[ f^2(x_t)\right]$.
\end{enumerate}
So $T_{\theta, y_{t}}( x_{t-1}, f)  \in L^2(\mathsf{X}, \phi_{\theta, y_{1:t-1} })$ and this gives us
\begin{equation*}
\mathbb{E}\left[  \exp\left( is N_X^{1/2} B_{N,\theta}  \right)  \right]  
\overset{\text{p}}{\to} 
\exp\left[ -\frac{s^2}{2}\sigma^2_{\theta,t-1}\left[ T_{\theta, y_{t}}  (x_{t-1},f)  \right] \right].
\end{equation*}

Now, focusing on $N_X^{1/2}A_{N,\theta}$, 
\begin{equation}\label{bnconvergence}
\mathbb{E}\left[ \exp\left( i r N_X^{1/2} A_{N,\theta} \right)  \mid \mathcal{F}_{\theta, t-1}^N  \right] \overset{\text{p}}{\to} \exp\left( - \frac{r^2}{2} \eta^2_{\theta,t-1}(f)  \right)
\end{equation}
by proposition~\ref{prop9512}, whose assumptions are now verified.

First, the triangular array $\{N_X^{-1/2} \tilde{W}_{t,\theta}^{N,j} f(\tilde{X}^{N,j}_t) \}_{1 \le j \le N_X}$ is conditionally independent given $\mathcal{F}_{\theta,t-1}^N$ by the description of the mutation step of algorithm~\ref{alg:sisr}.

Second, by assumption~\ref{assmp5}, the unnormalized weights are bounded above by some finite $M$, so
\begin{align}\label{eqn:upper_bnd}
\mathbb{E} \left[ (\tilde{W}_{t,\theta}^{N,j})^2 f^2(\tilde{X}^{N,j}_t) \mid \mathcal{F}_{\theta, t-1}^N \right] 
&= 
\int Q_{\theta,y_t}(x_{t-1}^{N,j}, d X_t) \left( \frac{dT(x_{t-1}^{N,j},\cdot)}{dQ(x_{t-1}^{N,j}, \cdot)}(\tilde{X}_t^{N,j}) \right)^2f^2(\tilde{X}_t^{N,j}) \\
&\le M  T_{\theta,y_t}(x_{t-1}^{N,j}, f^2);
\end{align}
the right hand side is finite with probability $1$ because we assume

\begin{equation}\label{eqn:fraction}
\phi_{\theta,y_{1:t}}(f^2) = \frac{\phi_{\theta,y_{1:t-1}}\left( T_{\theta,y_t}(x_{t-1}, f^2) \right)}{\phi_{\theta,y_{1:t-1}}\left( T_{\theta,yt}(x_{t-1}, 1) \right) } < \infty.
\end{equation}
Assumptions \ref{assmp2} and \ref{assmp3} imply that the numerator and denominator are finite, separately, so for any $N \ge 1$ and $1 \le j \le N_X$, $N_X^{-1} \mathbb{E}\left[ \left(\tilde{W}_{t,\theta}^{N,j}\right)^2 f^2(\tilde{X}^{N,j}_t) \mid \mathcal{F}_{\theta,t-1}^N \right]$ is finite. 

Third, we can show that 
\begin{equation}\label{expr:third_cond}
N_X^{-1} \sum_{j=1}^{N_X} \left\{ 
\mathbb{E}\left[ \left(\tilde{W}_{t,\theta}^{N,j}\right)^2 f^2(\tilde{X}_t^{N,j}) \mid \mathcal{F}_{\theta,t-1}^N \right]
-
\left( \mathbb{E}\left[ \tilde{W}_{t,\theta}^{N,j} f(\tilde{X}_t^{N,j}) \mid \mathcal{F}_{\theta,t-1}^N \right] \right)^2
\right\} 
\end{equation}
converges in probability to $\eta^2_{\theta,t-1}(f) >0$. The finiteness of equation~\ref{eqn:upper_bnd} and our assumption of consistency for $L^1(\mathsf{X}, \phi_{\theta,y_{1:t-1}})$ yield
\begin{equation*}
N_X^{-1} \sum_{j=1}^{N_X} \mathbb{E}\left[ \left(\tilde{W}_{t,\theta}^{N,j}\right)^2 f^2(\tilde{X}_t^{N,j}) \mid \mathcal{F}_{\theta,t-1}^N \right]
\overset{\text{p}}{\to}
\phi_{\theta,y_{1:t-1}} \left\{ \int Q_{\theta,y_t}(x_{t-1}, dx_t) \left( \frac{dT(x_{t-1}, \cdot)}{dQ(x_{t-1}, \cdot)}(x_t) \right)^2 f^2(x_t) \right\}.
\end{equation*}
Similarly, 

\begin{equation}
N_X^{-1} \sum_{j=1}^{N_X} \left( \mathbb{E}\left[ \tilde{W}_{t,\theta}^{N,j} f(\tilde{X}_t^{N,j}) \mid \mathcal{F}_{\theta,t-1}^N \right] \right)^2
\overset{\text{p}}{\to} \phi_{\theta,y_{1:t-1}} \left\{  T_{\theta,y_t}(x_{t-1},  f)^2 \right\},
\end{equation}
and so expression~\ref{expr:third_cond} does indeed converge to $\eta^2_{\theta,t-1}(f)$.

Regarding the fourth condition of proposition~\ref{prop9512}, pick any $\epsilon > 0$; for an arbitrary $C_{\epsilon}$ such that $N_X^{1/2} \epsilon \ge C_{\epsilon}$, we have

\begin{align*}
&N_X^{-1} \sum_{j=1}^{M_N} E\left[  \left(\tilde{W}_{t,\theta}^{N,j}\right)^2 f^2(\tilde{X}^{N,j}_t) \mathbf{1}\left(N_X^{-1/2} \left|\tilde{W}_{t,\theta}^{N,j}\right| \left| f(\tilde{X}^{N,j}_t) \right| \ge \epsilon \right) \mid \mathcal{F}^N_{\theta, t-1} \right] \\
&\le
N_X^{-1} \sum_{j=1}^{M_N} E\left[  \left(\tilde{W}_{t,\theta}^{N,j}\right)^2 f^2(\tilde{X}^{N,j}_t) \mathbf{1}\left( \left|\tilde{W}_{t,\theta}^{N,j}\right| \left| f(\tilde{X}^{N,j}_t) \right| \ge C_\epsilon \right) \mid \mathcal{F}^N_{\theta, t-1} \right].
\end{align*}
A dominated convergence argument similar to that used in the proof to lemma~\ref{lem:resampling} yields
\begin{align*}
0 &\le \text{plim}_{N_X \to \infty} N_X^{-1} \sum_{j=1}^{M_N} E\left[  \left(\tilde{W}_{t,\theta}^{N,j}\right)^2 f^2(\tilde{X}^{N,j}_t) \mathbf{1}\left(N_X^{-1/2} \left|\tilde{W}_{t,\theta}^{N,j}\right| \left| f(\tilde{X}^{N,j}_t) \right| \ge \epsilon \right) \mid \mathcal{F}^N_{\theta, t-1} \right] \\
&= \text{plim}_{C_\epsilon \to \infty} \text{plim}_{N_X \to \infty} N_X^{-1} \sum_{j=1}^{M_N} E\left[  \left(\tilde{W}_{t,\theta}^{N,j}\right)^2 f^2(\tilde{X}^{N,j}_t) \mathbf{1}\left(N_X^{-1/2} \left|\tilde{W}_{t,\theta}^{N,j}\right| \left| f(\tilde{X}^{N,j}_t) \right| \ge \epsilon \right) \mid \mathcal{F}^N_{\theta, t-1} \right] \\
&\le
\text{plim}_{C_\epsilon \to \infty} \phi_{\theta,y_{1:t-1}}\left\{ E\left[  \left(\tilde{W}_{t,\theta}^{N,1}\right)^2 f^2(\tilde{X}^{N,1}_t) \mathbf{1}\left( \left|\tilde{W}_{t,\theta}^{N,1}\right| \left| f(\tilde{X}^{N,1}_t) \right| \ge C_\epsilon \right) \mid \mathcal{F}^N_{\theta, t-1} \right]  \right\} \\
&\le
 \phi_{\theta,y_{1:t-1}}\left\{ \lim_{C_\epsilon \to \infty} E\left[  \left(\tilde{W}_{t,\theta}^{N,1}\right)^2 f^2(\tilde{X}^{N,1}_t) \mathbf{1}\left( \left|\tilde{W}_{t,\theta}^{N,1}\right| \left| f(\tilde{X}^{N,1}_t) \right| \ge C_\epsilon \right) \mid \mathcal{F}^N_{\theta, t-1} \right]  \right\} \\
&=  \phi_{\theta,y_{1:t-1}}(0) = 0.
\end{align*}

This dominated convergence theorem argument is available because these conditional expectations are bounded above the same expectation with the indicator function removed. Therefore, because of our assumption of consistency, proposition~\ref{prop9512} is applicable, and \ref{bnconvergence} holds.

Iterating the expectation of the joint characteristic function of $N_X^{1/2}A_{N,\theta}$ and $N_X^{1/2}B_{N,\theta}$, taking the limit as $N_X \to \infty$, and utilizing the dominated convergence theorem once again, we can see that these two pieces are asymptotically independent:
\begin{align*}
& \lim_{N\to \infty} \mathbb{E}\left( \exp\left( i r N_X^{1/2}A_{N,\theta} \right) \exp\left( is N_X^{1/2}B_{N,\theta}\right)  \right) \\
&= \mathbb{E}\left( \lim_{N\to \infty} \exp\left( is N_X^{1/2}B_{N,\theta}\right) \lim_{N\to \infty} \mathbb{E}\left[ \exp\left( i r N_X^{1/2}A_{N,\theta} \right)  \mid \mathcal{F}_{\theta, t-1}^N  \right] \right) \\
&= \mathbb{E}\left( \lim_{N\to \infty} \exp\left( is N_X^{1/2}B_{N,\theta}\right) \exp\left[ - \frac{r^2 }{2} \eta^2_{\theta,t-1}(f) \right] \right) \\
&= \lim_{N\to \infty} \mathbb{E}\left(  \exp\left( is N_X^{1/2}B_{N,\theta}\right)  \right) \exp\left[ - \frac{r^2 }{2} \eta^2_{\theta,t-1}(f) \right]  \\
&= \exp\left[ -\frac{s^2}{2}\sigma^2_{\theta,t-1}\left[ T_{\theta, y_{t}}  (x_{t-1},f)  \right] \right]\exp\left[ - \frac{r^2 }{2} \eta^2_{\theta,t-1}(f) \right].
\end{align*}

The delta method gives us
\begin{equation}
N_X^{1/2}\left\{A_{N,\theta} + B_{N,\theta}\right\} 
\overset{\text{D}}{\to} 
\mathcal{N}( 
0, 
\sigma^2_{\theta,t-1}\left[ T_{\theta, y_{t}}  (x_{t-1},f)  \right] + \eta^2_{\theta,t-1}(f)).
\end{equation}

Next, $\sum_{j=1}^{N_X} N_X^{-1} \tilde{W}_{t,\theta}^{N,j}$ converges in probability to $\phi_{\theta,y_{1:t-1}}(T_{\theta, y_{t}}( x_{t-1}, 1))$ because $T_{\theta, y_{t}}( x_{t-1}, 1)  \in L(\mathsf{X}, \phi_{\theta, y_{1:t-1} })$. Slutsky's theorem tells us that

\begin{equation}
\frac{ N_X^{1/2}\left\{A_{N,\theta} + B_{N,\theta}\right\}  }{ \sum_{j=1}^{N_X} N_X^{-1} \tilde{W}_{t,\theta}^{N,j} } \overset{\text{D}}{\to} 
\mathcal{N}\left(
0, 
\frac{ \sigma^2_{\theta,t-1}\left[ T_{\theta, y_{t}}  (x_{t-1},f)  \right] + \eta^2_{\theta,t-1}(f) }{ \left[ \phi_{\theta,y_{1:t-1}}(T(x_{t-1},1)) \right]^2 } \right). 
\end{equation}

If $\phi_{\theta,y_{1:t-1}}(T_{\theta,y_t}(x_{t-1}, f)) \neq 0$, perform the same proof on $\bar{f} := f - \phi_{\theta,y_{1:t-1}}(T_{\theta,y_t}(x_{t-1}, f))$.  All the same assumptions are met for this function, the expression for the asymptotic variance does not change, and 
\begin{equation*}
N_X^{1/2}\left\{ \hat{\phi}_{N_X, \theta,y_{1:t}}(f) - \phi_{\theta,y_{1:t}}( f) \right\} 
=
N_X^{1/2}\left\{ \hat{\phi}_{N_X, \theta,y_{1:t}}(\bar{f}) - \phi_{\theta,y_{1:t}}( \bar{f} ) \right\}.
\end{equation*}

\end{proof}

\begin{proof}[Proof of Lemma~\ref{lem:resampclt}]

Pick $\theta$, $t$, and $f \in L^2(\mathsf{X}, \phi_{\theta,y_{1:t}})$, and add and subtract $N_X^{1/2}\hat{\phi}_{N_X, \theta,y_{1:t}}(f)$ from the left hand side of the above expression. $N_X^{1/2} \left\{ \hat{\phi}_{N_X, \theta,y_{1:t}}(f) - \phi_{\theta,y_{1:t}}(f) \right\}$ is asymptotically normal by assumption, so we turn our attention to the difference of these two quantities.

The difference can be written in terms of conditional expectations as follows:
\begin{equation*}
N_X^{1/2} \left\{ \check{\phi}_{N_X, \theta,y_{1:t}}(f) - \hat{\phi}_{N_X, \theta,y_{1:t}}(f) \right\} 
=
N_X^{1/2} \left\{N_X^{-1} \sum_{j=1}^{N_X} f(X^{N,j}_{t}) - 
E\left[ f(X^{N,j}_t) \mid \tilde{\mathcal{F}}^N_{\theta, t} \right] \right\}.
\end{equation*}
This will converge to a mean $0$ normal distribution after we verify the assumptions of proposition~\ref{prop9512}.

The first assumption of this proposition is that the triangular array $\{N_X^{-1/2} f(X^{N,j}_t) \}_{1 \le j \le N_X}$ is conditionally independent given $\tilde{\mathcal{F}}^N_{\theta, t}$. This is true by the description of algorithm~\ref{alg:sisr}. 

Second, any conditional second moment 
$$
E \left[ f^2(X^{N,j}_t) \mid \tilde{\mathcal{F}}_{\theta, t}^N \right] 
= 
\sum_{j=1}^{N_X} \frac{\tilde{W}_{t,\theta^i}^{N,j} }{ \sum_{j'} \tilde{W}_{t,\theta^i}^{N,j'} } f^2(\tilde{X}^{N,j}_t)
\le \max_{1 \le j \le N_X} f^2(\tilde{X}^{N,j}_t)
$$
is finite for any fixed sample size because assumptions \ref{assmp4} and \ref{assmp5} guarantee the sum of normalized weights is $1$.

Third, we must show
\begin{align*}
\sum_{j=1}^{N_X} \text{Var}\left[ N_X^{-1/2}  f(X^{N,j}_t) \mid \tilde{\mathcal{F}}^N_{\theta, t} \right]
&= E\left[   f^2(X^{N,1}_t) \mid \tilde{\mathcal{F}}^N_{\theta, t} \right]
-
\left\{ E\left[   f(X^{N,1}_t) \mid \tilde{\mathcal{F}}^N_{\theta, t} \right] \right\}^2 \\
&= \hat{\phi}_{N_X, \theta,y_{1:t}}(f^2) - \left( \hat{\phi}_{N_X, \theta,y_{1:t}}(f) \right)^2
\end{align*}
converges in probability to some positive constant as $N_X \to \infty$. This is true by our assumption of the consistency of $\hat{\phi}_{N_X, \theta,y_{1:t}}(f'')$, and by the continuous mapping theorem.

Finally, regarding the fourth condition of proposition~\ref{prop9512}, pick any $\epsilon > 0$, and then pick $C_{\epsilon} > 0$. For $N$ is large enough so that $N_X^{1/2} \epsilon \ge C_{\epsilon}$, we have

\begin{align*}
&N_X^{-1} \sum_{j=1}^{M_N} E\left[   f^2(X^{N,j}_t) \mathbf{1}\left(N_X^{-1/2} \left| f(X^{N,j}_t) \right| \ge \epsilon \right) \mid \tilde{\mathcal{F}}^N_{\theta, t} \right] \\
&\le
N_X^{-1} \sum_{j=1}^{M_N} E\left[   f^2(X^{N,j}_t) \mathbf{1}\left(  \left| f(X^{N,j}_t) \right| \ge C_\epsilon \right) \mid \tilde{\mathcal{F}}^N_{\theta, t} \right].
\end{align*}
We also know 

\begin{equation*}
E\left[  f^2(X^{N,j}_t) \mathbf{1}\left(  \left| f(X^{N,j}_t) \right| \ge C_\epsilon \right) \mid \tilde{\mathcal{F}}^N_{\theta, t} \right] 
\le
E\left[   f^2(X^{N,j}_t)  \mid \tilde{\mathcal{F}}^N_{\theta, t} \right] < \infty.
\end{equation*}
Putting these two ideas together

\begin{align*}
0 &\le \text{plim}_{N_X \to \infty} N_X^{-1} \sum_{j=1}^{M_N} E\left[   f^2(X^{N,j}_t) \mathbf{1}\left(N_X^{-1/2} \left| f(X^{N,j}_t) \right| \ge \epsilon \right) \mid \tilde{\mathcal{F}}^N_{\theta, t} \right] \\
&= \text{plim}_{C_\epsilon \to \infty} \text{plim}_{N_X \to \infty} N_X^{-1} \sum_{j=1}^{M_N} E\left[ f^2(X^{N,j}_t) \mathbf{1}\left(N_X^{-1/2} \left| f(X^{N,j}_t) \right| \ge \epsilon \right) \mid \tilde{\mathcal{F}}^N_{\theta, t} \right] \\
&\le
\text{plim}_{C_\epsilon \to \infty} \phi_{\theta,y_{1:t-1}}\left\{ E\left[ f^2(X^{N,1}_t) \mathbf{1}\left( \left| f(X^{N,1}_t) \right| \ge C_\epsilon \right) \mid \tilde{\mathcal{F}}^N_{\theta, t} \right]  \right\} \\
&\le
 \phi_{\theta,y_{1:t-1}}\left\{ \lim_{C_\epsilon \to \infty} E\left[ f^2(X^{N,1}_t) \mathbf{1}\left( \left| f(X^{N,1}_t) \right| \ge C_\epsilon \right) \mid \tilde{\mathcal{F}}^N_{\theta, t} \right]  \right\} \tag{DCT} \\
&=  \phi_{\theta,y_{1:t-1}}(0) = 0.
\end{align*}

Proposition~\ref{prop9512} now yields
\begin{equation*}
E\left[ \exp\left( i r N_X^{1/2} \left\{ \check{\phi}_{N_X, \theta,y_{1:t}}(f) - \hat{\phi}_{N_X, \theta,y_{1:t}}(f) \right\} \right)  \mid \tilde{\mathcal{F}}^N_{\theta, t} \right]
\overset{\text{p}}{\to} 
\exp\left[ - \frac{r^2}{2} \left\{ \phi_{\theta,y_{1:t}}(f^2) - \left[ \phi_{\theta,y_{1:t}}(f) \right]^2  \right\}  \right]
\end{equation*}
as $N_X \to \infty$.

Finding the joint characteristic function can be done in the same way as in the previous section--by taking the limit, iterating the expectation, and using the dominated convergence theorem:
\begin{align*}
& \lim_{N_X \to \infty} E\left[ \exp\left( i r N_X^{1/2} \left\{ \check{\phi}_{N_X, \theta,y_{1:t}}(f) - \hat{\phi}_{N_X, \theta,y_{1:t}}(f) \right\} \right)\exp\left( is N_X^{1/2} \left\{ \hat{\phi}_{N_X, \theta,y_{1:t}}(f) - \phi_{\theta,y_{1:t}}(f) \right\}  \right) \right] \\
&=  
\exp\left[ - \frac{r^2}{2} \left\{ \phi_{\theta,y_{1:t}}(f^2) - \left[ \phi_{\theta,y_{1:t}}(f) \right]^2  \right\}  \right]
\exp\left[- \frac{s^2}{2} \tilde{\sigma}^2_{\theta,t}(f) \right]
\end{align*}
From this, the delta method gives us the asymptotic normality of $N_X^{1/2} \check{\phi}_{N_X, \theta,y_{1:t}}(f)$.
\end{proof}


\begin{proof}[Proof of Theorem~\ref{thm:sisrclt}]
Pick any $\theta \in \Theta$ and note that, for $f \in L^2(\mathsf{X}, \phi_{\theta,y_1})$, $N_X^{1/2}\hat{\phi}_{N_X, \theta,y_{1}}(f)$ is asymptotically normal by the traditional CLT and Slutsky's theorem. This is also a corollary of \cite[Theorem 9.1.8]{iihmm}, but in that case, the target measure is normalized.

For $\check{\phi}_{N_X, \theta,y_{1}}(f)$, obtain the result by applying Theorem~\ref{thm:sisr} and lemma~\ref{lem:resampclt}.

For $\check{\phi}_{N_X, \theta,y_{1:t}}(f)$  or  $\hat{\phi}_{N_X, \theta,y_{1:t}}(f)$ with $t \ge 2$, use induction with Theorem~\ref{thm:sisr} and lemmas \ref{lem:mutclt} and \ref{lem:resampclt}.

\end{proof}


\section*{References}

\bibliography{mybibfile}

\end{document}